\documentclass[a4paper,twoside, 11pt]{article}
\usepackage[margin=1in]{geometry}
\usepackage{amsmath,amsfonts,amsthm,fancyhdr,lipsum,amssymb} 
\usepackage{graphicx,cite,times}
\usepackage{enumitem,version}
\fancypagestyle{mystyle}{
	\fancyhf{}
	\fancyhead[EC]{Kalika Prasad, Hrishikesh Mahato  and Munesh Kumari}
	\fancyhead[OC]{A novel public key cryptography based on generalized Lucas matrices}
	\fancyfoot[C]{\thepage}
	
}
\usepackage{pdfpages,hyperref}

\newcommand{\Z}{\mathbb{Z}}
\newcommand{\N}{\mathbb{N}}

\newtheorem{theorem}{Theorem}[section]
\newtheorem{corollary}{Corollary}[theorem]
\newtheorem{lemma}[theorem]{Lemma}

\newtheorem{example}{Example}

\newtheorem{definition}{Definition}[section]

\numberwithin{equation}{section}
\Large
\begin{document}
	\title{A novel public key cryptography based on generalized Lucas matrices}
	\author{Kalika Prasad$^{1}$\footnote{E-mail: klkaprsd@gmail.com} 
	{\includegraphics[scale=.61]{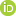}}, 
	Hrishikesh Mahato$^{2}$\footnote{E-mail: hrishikesh.mahato@cuj.ac.in} 
	{\includegraphics[scale=.61]{ORCID-1}} 
	and Munesh Kumari$^{3}$\footnote{E-mail: muneshnasir94@gmail.com}  
	{\includegraphics[scale=.61]{ORCID-1}}
	\\\normalsize{$^{1,2,3}$Department of Mathematics,
		Central University of Jharkhand, India, 835205}}     
	\date{\today}
	\maketitle
	\noindent\rule{16cm}{.15pt}
	\begin{abstract}
		In this article, we have proposed a generalized Lucas matrix (recursive matrix of higher order) having relation with generalized Fibonacci sequences and established many special properties in addition to that usual matrix algebra. Further, we have proposed a modified public key cryptography using these matrices as keys in Affine cipher and key agreement for encryption-decryption with the combination of terms of generalized Lucas sequences under residue operations. In this scheme, instead of exchanging the whole key matrix, only a pair of numbers(parameters) need to be exchanged, which reduces the time complexity as well as space complexity of the key transmission and has a large key-space.
	\end{abstract}
	\noindent\rule{16cm}{.1pt}
		
\section{Introduction}
	Matrix theory is rich in special properties, some of such properties are based on its construction, interrelations in eigen values and structure of its invertibility. Based on its construction, matrices are used in different branches of science as well as in engineering \& technologies. One of such area is cryptography\cite{stinson2005cryptography,stallings2017cryptography,paar2009understanding}, where matrix theory plays a vital role in storing data, efficiency of encryption-decryption and enlarging the key-spaces.
	
	It is well known that recursive sequences are defined in the terms of sum or difference or product(elementary operations) on preceding terms of corresponding sequences. Nowadays a lot of research work is going on in the direction of generalizing the existing sequences for higher order as well as generalizing for arbitrary initial values. While some of the authors made generalizations by considering the same relation but with different multipliers(constant/arbitrary functions as coefficients), some of such recent developments and their applications may be seen in	 \cite{tasci2004order,kilic2010generalized,bilgici2014two,kumari2020model,kumari2021role}.
	
	Akbary \& Wang\cite{akbary2006generalized}, proposed a variant of generalized  Lucas sequence and based on it constructed all the permutation binomials $P(x)=x^r+x^u$ of the finite Field $ \mathbb{F}_q $, where $ q=p^m $ with p is some odd prime and further they investigated permutations properties of a binomial.
	Cerda-Morales\cite{cerda2013generalized} defined a new generalized Lucas $V(p,q)$-matrix similar to the generalized Fibonacci $U(1,-1)$-matrix which is different from the Fibonacci $U(p,q)$-matrix and further they established some well-known equalities and a Binet-like formula by matrix method for both generalized sequences.
	Munesh et.al.\cite{kumari2021some} proposed k-Mersenne sequence and generalized	k-Gaussian Mersenne sequence which are equivalent to Fibonacci sequences and further they constructed recursive polynomials and established many well-known identities like Binet formula, Cassini’s identity, D’Ocagne’s identity, and generating functions. 
	Halici, et.al.\cite{halici2021fibonacci}, discussed the Fibonacci quaternion matrix by considering entries as n-th Fibonacci quaternion number and derived some identities like Cassini’s identity, Binet formula, ...etc.
	In\cite{stanimirovic2008generalization} Stanimirovic, et.al defined a generalization of Fibonacci and Lucas matrices whose elements are defined by the general second-order non-degenerated sequence and in some cases, they also obtained inverse for those matrices.
	
	E. {\"O}zkan, et.al\cite{ozkan2019generalized} obtained the terms of n-step Lucas polynomials by using matrices \& generalizing the concept and then establishing the relationship between Lucas polynomials and Fibonacci polynomials.
In\cite{singh2014generalized} Singh, et.al presented a generalization of Fibonacci-Lucas sequences, constructed matrix of order three and established some identities but they followed it for second order sequences, which can be further extended to higher order sequences and matrix construction.

We know that, well-known identities Fibonacci sequence \& Lucas sequence\cite{koshy2019fibonacci} is given by recurrence relation $f_{k+2} = f_{k} + f_{k+1}, ~~k\geq 0$ with initial values $ 0,1 $ and $ 2,1 $ respectively. Similarly, Tribonacci sequence \& Lucas sequence of order three is given by recurrence relation $f_{k+3} = f_{k} + f_{k+1}+f_{k+2}, ~~k\geq 0$ with initial values $ 0,0,1 $ [A000073] and $ 3,1,3 $ [A001644] respectively. And matrix representations\cite{koshy2019fibonacci} corresponding to above recursive sequences of order two and three has been obtained as follow, where $f_{k,n}$ represents $n^{th}$ term of the sequence of order k:
	\begin{eqnarray}
		\begin{bmatrix}	
			f_{2,n+1} & f_{2,n} \\ 
			f_{2,n} & f_{2,n-1} \\
		\end{bmatrix},
	~~\begin{bmatrix}
		f_{3,n+2} & f_{3,n+1} + f_{3,n} & f_{3,n+1}\\ 
		f_{3,n+1} & f_{3,n} + f_{3,n-1} & f_{3,n} \\
		f_{3,n} & f_{3,n-1} + f_{3,n-2} & f_{3,n-1} \\
	\end{bmatrix}	
	\end{eqnarray}
	These matrices are recursive in nature consisting of many properties based on initial values of corresponding sequences. For example, let us assume that $Q_k^n$ be the matrix of order $k$ denoting multiplication $Q_k$ to n times. Now, if we consider initial value $0,1$ for order two or $0,0,1$ for order three then $(Q_k^1)^n=Q_k^n$ holds but in other cases, it does not hold, some of such observations has been found in\cite{prasad2021cryptography,tianxiao2018matrix,yilmaz2013matrix,miles1960generalized,king1960some}.	
	
	 In this paper, we are working on generalizing the Lucas sequences to higher order preserving the Fibonacci trace properties(the terms of Lucas sequence are obtained by considering the trace of the corresponding Fibonacci matrices). The construction of recursive matrices has been done with entries from combinations of terms of proposed sequences and established a relationship with generalized Fibonacci matrices(GFM)\cite{prasad2021cryptography}. Further, we implement these matrices in the Affine-Hill technique and examine its behavior and strength.
	 
	 This paper is organized as follows. Section 1 introduces the related work on the construction of recursive matrices and their applications. In section 2, preliminaries on the cryptographic scheme, signature scheme \& their mathematical formulation are studied. In section 3, we have established generalized Lucas sequences \& matrices from GFM and discussed some remarkable properties. In section 4, we proposed a new algorithm for key exchange and encryption-decryption scheme with an example. Finally, in section 5 we discussed the implementation,  complexity and strength of the proposed scheme followed by a conclusion in section 6.
	 
\section{Encryption Scheme and Mathematical Flow}

	\subsection*{Affine-Hill Cipher}
	Affine-Hill Cipher is one of polygraphic block cipher equivalent to Hill Cipher\cite{stallings2017cryptography,stinson2005cryptography} which encrypts blocks of size $n \geq 1$ successive plaintext to ciphertext and vice-versa. Let us assume that $P$ be the plaintext $(p_1, p_2, ..., p_r)$, $ K $ be key matrix (simply key) and $C=(c_1, c_2, ..., c_r)$ be corresponding ciphertext of sizes $1\times rn$, $r\times r$ and $1\times rn$ respectively, where $p_i~\&~c_i$ are block vectors of size $1 \times n$.  i.e 
	\begin{eqnarray}
			P =
			\begin{bmatrix}	
				p_{1}& p_{2} & ... & p_{r}
			\end{bmatrix},
			K=
			\begin{bmatrix}	
				k_{11}& k_{12} & ... & k_{1n} \\ 		
				k_{21} & k_{22} & ... & k_{2n} \\ 		
				\vdots & \vdots & \ddots & \vdots \\ 		
				k_{n1} & k_{n2} & ... & k_{nn} 	
			\end{bmatrix}
			~~\text{and}~~
			C=
			\begin{bmatrix}	
				c_{1}& c_{2} & ... & c_{r}
			\end{bmatrix} \nonumber
	\end{eqnarray}
	The Affine-Hill Cipher's technique is described as:
	\begin{eqnarray}\label{Affine-Hill}
		Enc(P):\hspace{.7cm} c_{i} &\equiv& (p_{i}K + B) \pmod p \\
		Dec(C):\hspace{.7cm} p_{i} &\equiv& (c_{i} - B)K^{-1} \pmod p		
	\end{eqnarray}
	with $(|K|,p)=1$ where $B$ is a $1\times n$ row vector, $p$ is a prime number  and $Enc(P)$ represents encryption scheme \& $Dec(C)$ represents decryption scheme.
	\subsection{ElGamal  and Signature Scheme}\label{ElGamal}
	Elgamal cryptosystem\cite{stallings2017cryptography,stinson2005cryptography,paar2009understanding}, proposed by T. Elgamal\cite{elgamal1985public} in 1984 is a public-key scheme with digital signature whose strength is based on discrete logarithms (closely related to Deffie-Hellman technique).
	One of such similar digital signatures is 'Schnorr signature scheme'\cite{stallings2017cryptography} which minimizes the message based computation required to generate a signature. Usually, the digital signature scheme involves the use of a private key for the generation of the digital signature and a public key for its verification purpose.
	The design of the Elgamal technique is as encryption is done by user's public key while decryption using user's private key.
	\subsubsection*{Primitive root}
	Primitive roots\cite{stallings2017cryptography} play a crucial role in securing strength of cryptographic schemes. Let $ n $ be any positive integer, an integer $\alpha$ is called primitive root of $ n $ if $\alpha^k\equiv 1 \pmod n$ where $k=\phi{(n)}$ is least positive integer, i.e. there does not exists any $r, 1\leq r <k$ such that $\alpha^r\equiv 1 \pmod n$. Further,  when $\alpha$ is primitive root of $n$, the powers of $\alpha$, $\alpha, \alpha^2, \alpha^3, ..., \alpha^{\phi{(n)}}$ are distinct $ \pmod n $ and also co-prime to n. Note that not all integers have primitive roots, in fact integers of the form $2, 4, p^k~or~2p^k$ where $p$ is odd prime \& $k\in\N$ have primitive roots.
	
	Similar to the Diffie-Hellman scheme, in the Elgamal technique, the global elements are the prime $p$ and a primitive root of $p$.	Elgamal scheme is defined as:
	\subsubsection{Public key setup}\label{PublikKeySetup} Let us assume that $p$ be any odd prime number, $\alpha$ be a primitive root of $p$ and an integer D is chosen such that $1 < D < \phi(p)$. Now, assign $E_{1}=\alpha$ and $E_{2}= E_{1}^{D}\pmod p$.
	Then $pk(p, E_{1}, E_{2})$ will be made as a public key and chosen $D$ is kept as a secret key. Elgamal technique can be understood as follows. Suppose two parties Alice and Bob want to communicate with each other, then they go through the following path:	
	\subsubsection{Key exchanging}\label{EGkeyexchange}
	Alice(wish to send message to Bob) first choose a random integer $e$ such that $1 < e < \phi(p)$ and then generates a signature key (say $ s $) using public key $pk(p, E_{1}, E_{2})$ where $s = E_{1}^{e} \pmod p$. Now she computes her secret key $\lambda$ for encryption as $\lambda  = E_{2}^{e} \pmod p$.
	Thus Alice generates the parameters$(\lambda,s) $ of an encryption key using Bob's public key ($p, E_{1}, E_{2}$), then encrypt the plaintext with the encryption key and send $(\lambda, C)$ to Bob, where C is the corresponding ciphertext.
	\subsubsection{Key recover by Bob}
	On the other side Bob after receiving ($s, C)$ from Alice, recovers the same parameters$(\lambda,s) $ of encryption key using his secret key $ sk(D )$ as:
	\begin{eqnarray}\label{recoverkey}
		s^{D}\pmod p  &\equiv& (E_{1}^{e})^{D} \pmod p \nonumber\\
		&\equiv& (E_{1}^{D})^{e}\pmod p \equiv (E_{2})^{e}\pmod p = \lambda 
	\end{eqnarray}
	Thus, Bob and Alice agree on the same parameters$(\lambda,s) $ of encryption key i.e. the parameters$(\lambda,s)$ of encryption key exchanged securely and using these parameters, Bob can decrypt the ciphertext $C$, and recover the original plaintext $P$.
		
\section{Generalized Lucas Sequences and Matrix Construction}
	In this section, we discuss the construction of generalized Lucas sequences and generalized Lucas matrices(GLM). Further, we obtain some new properties for GLM and establish a relationship with GFM.
\subsection{Generalized Fibonacci Matrices(GFM)}
	For $n\in \Z$, the generalized Fibonacci sequences$\{f_{k,n}\}$ of order $k\geq 2$ is defined as follow:
	\begin{eqnarray}\label{genFib}
		f_{k,k+n} = f_{k,k+n-1}+f_{k,k+n-2}+f_{k,k+n-3}+...+f_{k,n+2}+f_{k,n+1}+f_{k,n} ~~~~ k(\geq 2)\in\N
	\end{eqnarray} where
	\begin{eqnarray}\label{genFibInitial}
		f_{k,0} =f_{k,1}=f_{k,2}=...=f_{k,k-2}=0~~\&~~ f_{k,k-1}=1
	\end{eqnarray}
	The corresponding generalized Fibonacci matrix($Q_k^n$) of order $k$ denoted as $GFM(k,n)$ is defined\cite{prasad2021cryptography} as\\
	$Q_k^n= 
	\begin{bmatrix}
		f_{k,k+n-1} & f_{k,k+n-2}+f_{k,k+n-3}+...+f_{k,n} & f_{k,k+n-2}+...+f_{k,n+1} & ... & f_{k,k+n-2} \\
		
		f_{k,k+n-2} & f_{k,k+n-3}+f_{k,k+n-4}+...+f_{k,n-1} & f_{k,k+n-3}+...+f_{k,n} & ... & f_{k,k+n-3} \\
		
		\vdots & \vdots & \vdots & \ddots & \vdots \\
		
		f_{k,k+n-(k-1)} & f_{k,n}+f_{k,n-1}+...+f_{k,-k+n+2} & f_{k,n}+...+f_{k,-k+n+3} & ... & f_{k,n} \\
		f_{k,n} & f_{k,n-1}+f_{k,n-2}+...+f_{k,-k+n+1} & f_{k,n-1}+...+f_{k,-k+n+2} & ... & f_{k,n-1} \\
	\end{bmatrix}$
	where initial matrices are given by  
	\begin{eqnarray}
		Q_{k}= Q_{k}^1 =
	\begin{bmatrix}	
		1 & 1 & 1 & ... & 1 & 1\\	
		1 & 0 & 0 & ... & 0 & 0\\
		0 & 1 & 0 & ... & 0 & 0\\
		\vdots & \vdots & \vdots & \ddots & \vdots \\
		0 & 0 & 0 & ... & 1 & 0\\
	\end{bmatrix}~~and~~
		Q_{k}^{-1} = 
	\begin{bmatrix}	
		0 & 1 & 0 & ... & 0 & 0\\	
		0 & 0 & 1 & ... & 0 & 0\\
		0 & 0 & 0 & ... & 0 & 0\\
		\vdots & \vdots & \vdots & \ddots & \vdots & \vdots \\
		0 & 0 & 0 & ... & 0 & 1\\
		1 & -1 & -1 & ... & -1 & -1\\
	\end{bmatrix}_{k\times k} \nonumber
	\end{eqnarray}
	\begin{theorem}\label{MultinacciProp}\cite{prasad2021cryptography}
		Let $k(\geq 2)\in\N,n\in\Z$, $Q_k^{n}$ be generalized Fibonacci matrices ($ GFM(k,n) $) and $I_{k}$ is identity matrix of order k then following properties holds for every $ n $ and $ k $:
		\begin{enumerate}
			\item ${(Q_k^{1})^n}=Q_k^{n}$ and ${(Q_k^{-1})^n}=Q_k^{-n}$.
			\item ${Q_k^{n}Q_k^{l}}=Q_k^{n+l}={Q_k^{l}Q_k^{n}}$.
			\item $Q_{k}^{n}Q_{k}^{-n} = Q_{k}^{0} = I_{k}$.
			\item $det(Q_{k}^{n}) = [(-1)^{k-1}]^n = (-1)^{(k-1)n}$.
		\end{enumerate}
	\end{theorem}
	\begin{lemma}\label{Q1Multiplication}
		Let $ Q_k^1$ be first Fibonacci matrix of order k and $ A=(a_{ij}) $ be any square matrix of same size, then on multiplication with $ Q_k^1$ to A,
		the first row of $ Q_k^1A$ will be the sum of corresponding columns of A and second row to $ k^{th} $-row of $ Q_k^1A$ becomes first to $ (k-1)^{kth} $ rows of A respectively.
	\end{lemma}
	\begin{proof}
		It can be easily proved by the usual matrix multiplication of $ Q_k^1$ and $A$.
	\end{proof}
	\subsection{Generalized Lucas Matrices(GLM)}
	Now, we are establishing a new sequence$\{l_{k,n}\} $ analogous to generalized Fibonacci sequence as follow:
	\begin{definition}
	Consider the recurrence relation $\{l_{k,n}\} $ of order $k\geq 2$ defined as 
	\begin{eqnarray}\label{genlucas}
		l_{k,k+n} = l_{k,k+n-1}+l_{k,k+n-2}+l_{k,k+n-3}+...+l_{k,n+2}+l_{k,n+1}+l_{k,n} ~~~~n\geq0, k(\geq 2)\in\N\\
		\text{with ${l_{k,r}} = trace(Q_k^r)$, $0\leq r<k$.}\nonumber
	\end{eqnarray}		
	\end{definition}
	 To obtain initial values for above proposed sequence, we are considering trace of first k generalized Fibonacci matrices. Using equation(\ref{genlucas}) and by the definition of $Q_K^n$ it can be observed that
	\begin{eqnarray}\label{tracerelation}
		l_{k,n} &=& trace(Q_k^n)~~~~~~~ for~n\in \Z,~k(\geq 2)\in \N \\
	i.e.~~~~l_{k,n}&=& f_{k,k+n-1}+1f_{k,k+n-3}+2f_{k,k+n-4}+3f_{k,k+n-5}+... +(k-3)f_{k,n+1} \nonumber\\
		& &+ (k-2)f_{k,n}+(k-1)f_{k,n-1}\nonumber\\
		&=& f_{k,k+n-1}+ \sum_{i=n-1}^{k+n-3}f_{k,i}+ \sum_{i=n-1}^{k+n-4}f_{k,i}+... + \sum_{i=n-1}^{n-2}f_{k,i}+ \sum_{i=n-1}^{n-1}f_{k,i}
	\end{eqnarray}
	which yields $k,1,3,7,15,31,63,127,255,511,..., 2^{k-1}-1$ as initial values of $\{l_{k,n}\} $.
	
	Since the generalized Fibonacci sequences $\{f_{k,n}\}$ and GFM(k,n) $ \{Q_k^n\} $ are both two ended sequences, hence $\{l_{k,n}\} $ can also be extended in negative direction.
	The new sequence $\{l_{k,n} \}$ is called the generalized Lucas sequence of order k. In particular, $ k=2 $ gives the standard Lucas sequence of order 2.

	Here, the recursive matrix($L_k^{(n)}$) corresponding to the above sequence is defined as
	\begin{eqnarray}\label{GenLucasMatrix}
		L_k^{(n)}= 
		\begin{bmatrix}
			l_{k,k+n-1} & l_{k,k+n-2}+l_{k,k+n-3}+...+l_{k,n} & l_{k,k+n-2}+...+l_{k,n+1} & ... & l_{k,k+n-2} \\
			
			l_{k,k+n-2} & l_{k,k+n-3}+l_{k,k+n-4}+...+l_{k,n-1} & l_{k,k+n-3}+...+l_{k,n} & ... & l_{k,k+n-3} \\
			
			\vdots & \vdots & \vdots & \ddots & \vdots \\
			
			l_{k,k+n-(k-1)} & l_{k,n}+l_{k,n-1}+...+l_{k,-k+n+2} & l_{k,n}+...+l_{k,-k+n+3} & ... & l_{k,n} \\
			l_{k,k+n-k} & l_{k,n-1}+l_{k,n-2}+...+l_{k,-k+n+1} & l_{k,n-1}+...+l_{k,-k+n+2} & ... & l_{k,n-1} \\
		\end{bmatrix}.
	\end{eqnarray}
	And the matrix($L_k^{(n)}$) is known as generalized Lucas Matrix of order $ k $ and we denote it by GLM(k,n).
	Further, the initial Lucas matrix $L_k^{(0)}$ is
	\begin{eqnarray}\label{GenInitalLucasMatrix}
		L_k^{(0)} &=&
		\begin{bmatrix}
			2^{k-1}-1 & 2^{k-1} & 2^{k-1}-k & ... & 7.2^{k-4}& 3.2^{k-3} & 2^{k-2}-1\\
			
			2^{k-2}-1 & 2^{k-2} & 2^{k-2}+1 & ... & 7.2^{k-5}& 3.2^{k-4} & 2^{k-3}-1 \\
			2^{k-3}-1 & 2^{k-3} & 2^{k-3}+1& ... & 7.2^{k-6}& 3.2^{k-5} & 2^{k-4}-1 \\
			\vdots & \vdots & \vdots & \ddots & \vdots \\
			1 & 2 & 3 & ...& k-2 & k-1 & k \\
			k & 1-k & 2-k & ...& -3 & -2 & -1 \\
		\end{bmatrix}
	\end{eqnarray}	
	\begin{example}
	Initial Lucas matrix of order two, three, four and five are respectively
	\begin{eqnarray}
		\begin{bmatrix}
			1 & 2\\
			2 & -1\\
		\end{bmatrix},
		\begin{bmatrix}
			3 & 4 & 1\\	
			1 & 2 & 3\\
			3 & -2 & -1\\
		\end{bmatrix},
		\begin{bmatrix}
			7 & 8 & 4 & 3\\	
			3 & 4 & 5 & 1\\	
			1 & 2 & 3 & 4\\				
			4 & -3 & -2 & -1\\	
		\end{bmatrix}
	~and~
	\begin{bmatrix}
		15 & 16 & 11 & 10 & 7\\
		7 & 8 & 9 & 4 & 3\\	
		3 & 4 & 5 & 6&  1\\	
		1 & 2 & 3 & 4 & 5\\				
		5 & -4 & -3 & -2 & -1\\	
	\end{bmatrix}\nonumber	
	\end{eqnarray}
	\end{example}
	\begin{theorem}\label{Fibo=Lucas}
		Let  $n\in\Z$, $Q_k^{n}$ be generalized Fibonacci matrix, $L_{k}^{(n)}$ be  $ GLM(k,n) $. Suppose $L_{k}^{(0)}$ be the initial matrix as defined by equation(\ref{GenInitalLucasMatrix}), then we have
		\begin{eqnarray}
			{L_k^{(n)}} = Q_{k}^{n}L_{k}^{(0)}=L_{k}^{(0)} Q_{k}^{n}  
		\end{eqnarray}
	\end{theorem}
	\begin{proof}
		We prove it using mathematical induction on n, for $n\geq 1$. Since for $ n=0 $, result holds obviously. For $ n=1 $, we have
		\begin{eqnarray}
		Q_{k}^{1}L_{k}^{(0)} &=&
		\begin{bmatrix}	
			1 & 1 & 1 & ... & 1 & 1\\	
			1 & 0 & 0 & ... & 0 & 0\\
			0 & 1 & 0 & ... & 0 & 0\\
			\vdots & \vdots & \vdots & \ddots & \vdots \\
			0 & 0 & 0 & ... & 0 & 0\\
			0 & 0 & 0 & ... & 1 & 0\\
		\end{bmatrix}
		\begin{bmatrix}
			2^{k-1}-1 & 2^{k-1} & 2^{k-1}-k & ... & 7.2^{k-4}& 3.2^{k-3} & 2^{k-2}-1\\
			
			2^{k-2}-1 & 2^{k-2} & 2^{k-2}+1 & ... & 7.2^{k-5}& 3.2^{k-4} & 2^{k-3}-1 \\
			2^{k-3}-1 & 2^{k-3} & 2^{k-3}+1& ... & 7.2^{k-6}& 3.2^{k-5} & 2^{k-4}-1 \\
			\vdots & \vdots & \vdots & \ddots & \vdots \\
			1 & 2 & 3 & ...& k-2 & k-1 & k \\
			k & 1-k & 2-k & ...& -3 & -2 & -1 \\
		\end{bmatrix} \nonumber\\
	& & \text{using lemma(\ref{Q1Multiplication})}\nonumber\\
		&=&
		\begin{bmatrix}
			2^{k}-1 & 2^{k}-k-1 & 2^{k}-2-k& ... & 7.2^{k-3}& 3.2^{k-2} & 2^{k-1}-1 \\
			2^{k-1}-1 & 2^{k-1} & 2^{k-1}-k & ... & 7.2^{k-4}& 3.2^{k-3} & 2^{k-2}-1\\
			
			2^{k-2}-1 & 2^{k-2} & 2^{k-2}+1 & ... & 7.2^{k-5}& 3.2^{k-4} & 2^{k-3}-1 \\
			\vdots & \vdots & \vdots & \ddots & \vdots \\
			1 & 2 & 3 & ...& k-2 & k-1 & k \\
		\end{bmatrix} = L_{k}^{(1)} \nonumber
		\end{eqnarray}
	Now, assuming the result is true for $n=r$ i.e. $Q_{k}^{r}L_{k}^{(0)} = L_k^{(r)}$. Further we show that the statement is also true for $n=r+1$. We have, 
	\begin{eqnarray}
		Q_{k}^{r+1}L_{k}^{(0)} &=& Q_{k}^{1}Q_{k}^{r}L_{k}^{(0)} = Q_{k}^{1}L_{k}^{(r)}\nonumber\\
		&=& Q_{k}^{1} 
		\begin{bmatrix}
			l_{k,k+r-1} & l_{k,k+r-2}+l_{k,k+r-3}+...+l_{k,r} & l_{k,k+r-2}+...+l_{k,r+1} & ... & l_{k,k+r-2} \\
			
			l_{k,k+r-2} & l_{k,k+r-3}+l_{k,k+r-4}+...+l_{k,r-1} & l_{k,k+r-3}+...+l_{k,r} & ... & l_{k,k+r-3} \\
			
			\vdots & \vdots & \vdots & \ddots & \vdots \\
			
			l_{k,k+r-(k-1)} & l_{k,r}+l_{k,r-1}+...+l_{k,-k+r+2} & l_{k,r}+...+l_{k,-k+r+3} & ... & l_{k,r} \\
			l_{k,k+r-k} & l_{k,r-1}+l_{k,r-2}+...+l_{k,-k+r+1} & l_{k,r-1}+...+l_{k,-k+r+2} & ... & l_{k,r-1} \\
		\end{bmatrix} \nonumber\\
	& &\text{using lemma(\ref{Q1Multiplication})}\nonumber\\
	&=& ~~~~~
	\begin{bmatrix}		
		l_{k,k+r} & l_{k,k+r-1}+l_{k,k+r-2}+...+l_{k,r+1} & l_{k,k+r-1}+...+l_{k,r+2} & ... & l_{k,k+r-1} \\
		l_{k,k+r-1} & l_{k,k+r-2}+l_{k,k+r-3}+...+l_{k,r} & l_{k,k+r-2}+...+l_{k,r+1} & ... & l_{k,k+r-2} \\
		
		l_{k,k+r-2} & l_{k,k+r-3}+l_{k,k+r-4}+...+l_{k,r-1} & l_{k,k+r-3}+...+l_{k,r} & ... & l_{k,k+r-3} \\
		
		\vdots & \vdots & \vdots & \ddots & \vdots \\
		
		l_{k,k+r-(k-1)} & l_{k,r}+l_{k,r-1}+...+l_{k,-k+r+2} & l_{k,r}+...+l_{k,-k+r+3} & ... & l_{k,r} \\
	\end{bmatrix} \nonumber\\
	&=& ~~~~~L_{k}^{(r+1)}.\nonumber
	\end{eqnarray}
	For negative direction, consider $n=-r$, where $r\geq 0$. Then by mathematical induction on $r$, it can be proved for negative direction with a similar argument. Combining both the cases, the result holds for all $n\in \mathbb{Z}$. The second equality can be proved similarly as first.
	\end{proof}
	\begin{theorem}\label{LucasGen}
		Let $~n\in\N$, then we have~~
		$(L_{k}^{(1)})^n = (L_{k}^{(n)})(L_{k}^{(0)})^{n-1}$.
	\end{theorem}
	\begin{proof}
	\begin{eqnarray}
		(L_{k}^{(1)})^n = (Q_{k}^{1}L_{k}^{(0)})^n &=& Q_{k}^{1}L_{k}^{(0)} Q_{k}^{1}L_{k}^{(0)}...n-times...Q_{k}^{1}L_{k}^{(0)}
		= (Q_{k}^{1})^n(L_{k}^{(0)})^n \nonumber\\
		&=& (Q_{k}^{n})(L_{k}^{(0)})^n = (Q_{k}^{n}L_{k}^{(0)})(L_{k}^{(0)})^{n-1} = (L_{k}^{(n)})(L_{k}^{(0)})^{n-1}
	\end{eqnarray}
	\end{proof}
	\begin{theorem}\label{GLMdet}
		Determinant of generalized Lucas matrices are given by
		\begin{eqnarray}
			det(L_{k}^{(n)}) =  \left\{ 
			\begin{array}{lr} 
				det(L_{k}^{(0)}) & : if~k=odd~~\\
				(-1)^n det(L_{k}^{(0)}) & :if~k=even
			\end{array}
			\right.
		\end{eqnarray}
	\end{theorem}
	\begin{proof}
		We have 
		\begin{eqnarray}
			det(L_{k}^{(n)}) = det(Q_{k}^{(n)}L_{k}^{(0)}) 
		&=& det(Q_{k}^{n})det(L_{k}^{(0)}) = (-1)^{(k-1)n}det(L_{k}^{(0)})\nonumber\\
		 &=& \left\{ 
		\begin{array}{lr} 
			det(L_{k}^{(0)}) & : if~k=odd~~\\
			(-1)^n det(L_{k}^{(0)}) & :if~k=even
		\end{array}
		\right.
		\end{eqnarray}
	\end{proof}

	\begin{theorem}
	For $ m,n\in\Z,~k(\geq2)\in\N $, we have
	\begin{eqnarray}
		{L_k^{(m)}L_k^{(n)}}= L_k^{(m+n)} L_{k}^{(0)} = {L_k^{(n)}L_k^{(m)}}
	\end{eqnarray}
	\end{theorem}
	\begin{proof}
		From theorem(\ref{Fibo=Lucas}) and theorem(\ref{MultinacciProp}) we have,
		\begin{eqnarray}
			L_k^{(m)}L_k^{(n)} = Q_{k}^{(m)}L_{k}^{(0)}Q_{k}^{(n)}L_{k}^{(0)}
			= Q_{k}^{(m)}Q_{k}^{(n)}L_{k}^{(0)}L_{k}^{(0)} = Q_{k}^{(m+n)}L_{k}^{(0)}L_{k}^{(0)}
			= L_k^{(m+n)} L_{k}^{(0)}\nonumber\\			
			L_k^{(n)}L_k^{(m)} = Q_{k}^{(n)}L_{k}^{(0)}Q_{k}^{(m)}L_{k}^{(0)}
			= Q_{k}^{(n)}Q_{k}^{(m)}L_{k}^{(0)}L_{k}^{(0)} = Q_{k}^{(m+n)}L_{k}^{(0)}L_{k}^{(0)}
			= L_k^{(m+n)} L_{k}^{(0)}\nonumber
		\end{eqnarray}
	\end{proof}
	
	\begin{corollary}
		For $ k(\geq2)\in \N,~n\in\Z $, we have~~
		$L_{k}^{(n)}L_{k}^{(0)} = L_{k}^{(0)}L_{k}^{(n)}$.
	\end{corollary}
	\begin{theorem}\label{inv=rel}
		Let $L_{k}^{(n)}$ be the ~$k^{th}$ order generalized Lucas matrix and $Q_{k}^{n}$ be GFM(k,n), then
		\begin{eqnarray}
			L_{k}^{(n)}L_{k}^{(-n)} = H \hspace{.4cm}~~\forall ~~n\in\Z,~~where ~H=(L_{k}^{(0)})^{2}
		\end{eqnarray}
	\end{theorem}
	\begin{proof}
		For $ n\in \N $, we  have
		\begin{eqnarray}
			L_{k}^{(n)}L_{k}^{(-n)} &=& Q_{k}^{n}L_{k}^{(0)}Q_{k}^{-n}L_{k}^{(0)}
			= Q_{k}^{n}Q_{k}^{-n}L_{k}^{(0)}L_{k}^{(0)}= I_kH=H\nonumber
		\end{eqnarray}
	\end{proof}
	\begin{theorem}\label{thminvofLucas}
		Suppose $L_{k}^{(n)}$ be a GLM(k,n) and $H=(L_{k}^{(0)})^{2}$ is invertible, then inverse is given by
		$$Inv(L_{k}^{(n)}) = L_{k}^{(-n)}H^{-1}$$
	\end{theorem}
	\begin{proof}
		Since, we have $L_{k}^{(n)}L_{k}^{(-n)} = H$ and $(L_{k}^{(0)})$ is invertible matrix hence does H. Now, on post multiplication with $H^{-1}$, we get
 		$$L_{k}^{(n)}(L_{k}^{(-n)}H^{-1}) =I_{k}; ~~~~~I_{k}~\text{is identity matrix.}$$
		Thus, for every GLM(k) $L_{k}^{(n)}$ there is a unique matrix $(L_{k}^{(-n)}H^{-1})$ such that their product is an identity matrix. Hence
		$$Inv(L_{k}^{(n)}) = L_{k}^{(-n)}H^{-1}$$
	\end{proof}
	\begin{lemma}
		Let p be a prime number and L is a generalized Lucas matrix, then we have 
		\begin{eqnarray}
		det(L)\pmod p = det(L\pmod p)
		\end{eqnarray}
	\end{lemma}
	\begin{theorem}\cite{stallings2017cryptography}\label{entryModulo}
		Let $ A=(a_{ij})$ be any matrix, then we have
		$$A \pmod{p} = [a_{ij}\pmod{p}]$$
	\end{theorem}


\section{Encryption Scheme \& Algorithm}

	Let us assume that receiver's public key is $pk(p, E_{1},E_{2})$ whose components are constructed by receiver(Bob)  as discussed in subsection(\ref{PublikKeySetup}). Now, a sender(Alice) constructs a secret key (say $\lambda$) with this public key by choosing a restricted integer and form an encryption matrix with their signature. Further, after receiving the encrypted message with the signature from Alice, Bob retrieve the secret key($\lambda$) and after some calculation recover the plain text (see section  \ref{ElGamal}). The following algorithm summarizes the methodology.
\subsection{Algorithm}\label{algorithm}
\subsubsection*{Encryption Algorithm(sender have access to $pk(p, E_{1},E_{2})$):}
	\begin{enumerate}
		\item Sender(Alice) first chooses a secret number $e$, such that $1 < e < \phi(p)$.
		\item \textbf{Signature:} $s \leftarrow E_{1}^{e}\pmod{p}$.
		\item \textbf{Secret key:} $\lambda \leftarrow E_{2}^{e}\pmod{p}$.
		\item Initiate Lucas sequence$\{l_{\lambda,s}\}$ of order $\lambda$.
		\item \textbf{Key Matrix:} $K \leftarrow L_{\lambda}^{(s)} \pmod{p}$, where $L_{\lambda}^{(s)}$ may be obtained from step-4 and equ.{(\ref{GenLucasMatrix})}.
		\item \textbf{Shift vector:} $B \leftarrow [l_{\lambda,\lambda}, l_{\lambda,\lambda+1}, ...,l_{\lambda,2\lambda-1}]$ of size $1\times\lambda$ (using step-4).
		\item  \textbf{Encryption:} $C=Enc(P):~~~~c_{i}\leftarrow (p_{i}K + B) \pmod{p}$.
		\item  transmit $(C,s)$ to Bob.
	\end{enumerate}
    \textbf{Decryption Algorithm:} After receiving $( C,s)$ from sender.
	\begin{enumerate}
		\item \textbf{Secret key:} $\lambda\leftarrow s^{D}\pmod{p}$, where $D$ is Bob's chosen secret key. 
		\item Initiate Lucas sequence$\{l_{\lambda,s}\}$.
		\item \textbf{Key Matrix:} $K^* \leftarrow L_{\lambda}^{(-s)}H^{-1}\pmod{p}$, where both $L_{\lambda}^{(-s)}$ \& $H=(L_{\lambda}^{(0)})^2$ may be obtained from step-2 using eqn.({\ref{GenLucasMatrix}}) \& eqn.(\ref{GenInitalLucasMatrix} respectively.
		\item \textbf{Shift vector:} $B \leftarrow [l_{\lambda,\lambda}, l_{\lambda,\lambda+1}, ...,l_{\lambda,2\lambda-1}]$ of size $1\times\lambda$ (using step-2).
		\item \textbf{Decryption:} $P=Dec(C):~~p_{i} \leftarrow(c_{i} - B)K^{*}\pmod{p}$
		\item Plaintext(P) recovered.
	\end{enumerate}
	\subsection{Example}
	 \begin{example}\label{eg1}
	 	Let us assume that Alice(sender) wish to send message to Bob(receiver). Consider the prime number $ p=37$. Establish the public key and secret key of  communication  for Bob. 
	 \end{example}
 	\begin{proof}[Solution]
 		Bob first choose an integer D is such that $1< D< \phi(37)=36$, say $ D=10 $. The set of primitive roots of 37 is given by X= \{2, 5, 13, 15, 17, 18, 19, 20, 22, 24, 32, 35 \}. Now Bob selects a primitive root $\alpha=17$ of $ p $ from X.
 		Now, according to public key setup(\ref{PublikKeySetup}), Bob assigns $E_1=17, E_2=E_1^D\pmod{p}\equiv 17^{10}\pmod{37} \equiv 28$.
 		Thus public key $ pk(p, E_{1},E_{2})$ for Bob is $ pk(37,17,28) $ and secret key is $ sk(10) $. Now using $ pk(37,17,28) $ anyone can send message to Bob (explained in next example).
 	\end{proof}
	\begin{example}[Encryption-Decryption]
	Using $ pk(37,17,28)$, construct the key matrix \& shift vector and encrypt the plaintext \textbf{NOBLE2022}.
	\end{example}
	\begin{proof}[Solution]
	Here the numerical values equivalent to \textbf{NOBLE2022} is $[13,14,01,11,04,28,26,28,28]$. Let us consider the alphabets $\Sigma =\Z_{37}$ defined as for letters from A-Z equivalent to 00-25, digits 0-9 are that to 26-35 and 36 for blank/white space.	
	Now, according the algorithm(\ref{algorithm}),\\
	Alice first choose an integer $ e $ such that $1 < e < \phi(37)$, say $ e=23 $.\\
	Now, Alice make signature(s) as $s=E_{1}^{e} = 17^{23}\pmod{37}\equiv 18$.
	\\Thus secret key for Alice is $\lambda = E_{2}^{e} = 28^{23}\pmod{37} \equiv 3$ and the key matrix $K=L_{\lambda}^{(s)}\pmod{p}$ for encryption is given by
	\begin{eqnarray}
		K = L_{3}^{(18)} &=& 
		\begin{bmatrix}\label{Q2}
		l_{3,20} & l_{3,19}+l_{3,18} & l_{3,19} \\			
		l_{3,19} & l_{3,18}+l_{3,17} & l_{3,18} \\		
		l_{3,18} & l_{3,17}+l_{3,16} & l_{3,17} \\			
	\end{bmatrix}\pmod{37}\nonumber\\
	 &=& \begin{bmatrix}
		196331 & 164778 & 106743 \\	
		106743 & 89588 & 58035 \\	
		58035 & 48708 & 31553 \\	
	\end{bmatrix}\pmod{37}~
	 = 
	\begin{bmatrix}
		9 & 17 & 35 \\	
		35 & 11 & 19 \\	
		19 & 16 & 29 \\
	\end{bmatrix}
	\end{eqnarray}
	which is obtained by substituting the values of corresponding terms of 
	generalized Lucas sequence for $\lambda = 3$  as given in the table
	\begin{table}[h]\label{Lucas3table}
		{\resizebox{\linewidth}{!}
	{\begin{tabular}{|c|c|c|c|c|c|c|c|c|c|c|c|c|c|c|c|c|c|c|c|c|c|c|c|c|c|c|}
			\hline 
		index (n)&... & -1 & 0 & 1 & 2 & 3 & 4 & 5 & 6 & ...&15 &16&17&18&19&20&... \\ 
			\hline 
		Lucas Seq.($ l_{3,n}$) & ...&  -1 & \textbf{3} & \textbf{1} & \textbf{3} & 7 & 11 & 21 & 39 &  ...&9327 &17155 &31553 & 58035 & 106743 &196331&... \\ 
			\hline
	\end{tabular}}}
	\caption{Lucas sequence of order 3}
	\end{table}
	\\And shift vector(B) is, $ B=[l_{3,3},l_{3,4},l_{3,5}]=[07,11,21] $.
	\\Now divide the plaintext \textbf{NOBLE2022} in blocks of size $1\times \lambda$  as follow:
	\\$P_{1}=[~N~O ~B]=[13~~14~~01], P_{2}=[L ~E ~2]=[11~04~28] ~~\text{and}~~P_{3}=[0~2~2]= [26~28~28]$.
	\\\textbf{Encryption takes places as:}
	$C_i \leftarrow  (P_iK + B)\pmod{37}$.
	\\ $C_{1} = (P_{1}K+B)\equiv$ 
	$\left(\begin{bmatrix}
		13 & 14 & 01 \\
	\end{bmatrix}	
	\\\begin{bmatrix}
		9 & 17 & 35 \\	
		35 & 11 & 19 \\	
		19 & 16 & 29 \\	
	\end{bmatrix}+
	\begin{bmatrix}
		07 & 11 & 21 \\
	\end{bmatrix}\right)\pmod{37}$  	\\\indent$\equiv (04~~32~~31) \sim$ (E 7 6)
	\\ $C_{2} = (P_{2}K+B)\equiv$ 
	$\left(\begin{bmatrix}
		11 & 04 & 28 \\
	\end{bmatrix}	
	\\\begin{bmatrix}
		9 & 17 & 35 \\	
		35 & 11 & 19 \\	
		19 & 16 & 29 \\
	\end{bmatrix}+
	\begin{bmatrix}
		07 & 11 & 21 \\
	\end{bmatrix}\right)\pmod{37}$		\\\indent$\equiv (01~~24~~36) \sim$ (B Y $ \square $)
	\\ $C_{3} = (P_{3}K+B)\equiv$ 
	$\left(\begin{bmatrix}
		26 & 28 & 28 \\
	\end{bmatrix}	
	\\\begin{bmatrix}
		9 & 17 & 35 \\	
		35 & 11 & 19 \\	
		19 & 16 & 29 \\
	\end{bmatrix}+
	\begin{bmatrix}
		07 & 11 & 21 \\
	\end{bmatrix}\right)\pmod{37}$		\\\indent$\equiv (14~~25~~18) \sim$ (O Z S)
	\\Thus, Alice encrypted the plaintext \textbf{NOBLE2022} to \textbf{E76BY$\square$OZS},
	and send it to Bob along with her signature i.e Alice sends $\{s=18, C=C_1C_2C_3\}$ to Bob.
	\\\textbf{Decryption:} On the other side 'Bob' receives $(C,s)$ from Alice. To construct the decryption key $K^{*}$ Bob first recovers $\lambda$. It can be obtained using their secret key $D$ as follow:
	\begin{eqnarray}
		\lambda = s^{D}\pmod{37}= 18^{10}\pmod{37} \equiv 3. \nonumber
	\end{eqnarray}
	Thus $K^* = L_{3}^{(-s)}H^{-1}\pmod{p}$, where $H=(L_{3}^{(0)})^2 = \begin{bmatrix}
		16 & 18 & 14 \\		
		14 & 2 & 4 \\		
		4 & 10 & -2 \\		
	\end{bmatrix}$ is given as
	\\\begin{eqnarray}
		K^{*} &=& L_{3}^{(-18)}H^{-1}\pmod{37} \nonumber\\
		&=& 
		\begin{bmatrix}
			l_{3,-16} & l_{3,-17}+l_{3,-18} & l_{3,-17} \\			
			l_{3,-17} & l_{3,-18}+l_{3,-19} & l_{3,-18} \\		
			l_{3,-18} & l_{3,-19}+l_{3,-20} & l_{3,-19} \\		
		\end{bmatrix} \dfrac{1}{44}
		\begin{bmatrix}
			-1 & 4 & 1 \\		
			1 & -2 & 3 \\		
			3 & -2 & -5 \\		
		\end{bmatrix} \pmod{37}\nonumber\\
		&=&
		\dfrac{1}{44}
		\begin{bmatrix}
			-253 & 318 & 271 \\		
			271 & -524 & 47 \\		
			47 & -224 & -571 \\		
		\end{bmatrix}
		\begin{bmatrix}
			-1 & 4 & 1 \\		
			1 & -2 & 3 \\		
			3 & -2 & -5 \\		
		\end{bmatrix} \pmod{37}=
		\begin{bmatrix}
			18 & 36 & 7 \\		
			7 & 11 & 29 \\		
			29 & 15 & 19 \\		
		\end{bmatrix}\nonumber
	\end{eqnarray}
	Clearly, $ K.K^{*} = I\pmod{37}$ (by the theorem (\ref{thminvofLucas})).
	\\And shift vector is recovered by $\lambda$ as $ B=[l_{3,3},l_{3,4},l_{3,5}]=[07,11,21]  $.
	Note that, entries of both $L_{3}^{-18}$ and H may be obtained from the table(\ref{Lucas3table})
	\\Thus plaintext may be obtained by $P_{i}\leftarrow (C_{i} - B).K^{*}\pmod{37}$.
	\\$P_{1} = (C_{1}-B)K^{*} \equiv$ 
	$\left(\begin{bmatrix}
		04 & 32 & 31 \\ 
	\end{bmatrix}	
	-
	\begin{bmatrix}
		07 & 11 & 21 \\
	\end{bmatrix}\right)
	\begin{bmatrix}
		18 & 36 & 7 \\		
		7 & 11 & 29 \\		
		29 & 15 & 19 \\		
	\end{bmatrix}\pmod{37} \\ \indent \equiv (13~~14~~01) \sim$ (N O B)
	\\
	$P_{2} = (C_{2}-B)K^{*} \equiv$ 
	$\left(\begin{bmatrix}
		01 & 24 & 36 \\
	\end{bmatrix}	
	-
	\begin{bmatrix}
		07 & 11 & 21 \\
	\end{bmatrix}\right)
	\begin{bmatrix}
		18 & 36 & 7 \\		
		7 & 11 & 29 \\		
		29 & 15 & 19 \\	
	\end{bmatrix}\pmod{37} \\ \indent \equiv (11~~04~~28) \sim$ (L E 2)
	\\$P_{3} = (C_{3}-B)K^{*} \equiv$ 
	$\left(\begin{bmatrix}
		14 & 25 & 18 \\
	\end{bmatrix}	
	-
	\begin{bmatrix}
		07 & 11 & 21 \\
	\end{bmatrix}\right)
	\begin{bmatrix}
		18 & 36 & 7 \\		
		7 & 11 & 29 \\		
		29 & 15 & 19 \\		
	\end{bmatrix}\pmod{37} \\ \indent \equiv (26~~28~~28) \sim$ (0 2 2)
	\\Thus, the plaintext \textbf{NOBLE2022} successfully received by Bob.
	\end{proof}
\section{Strength-Analysis}
	In our proposed scheme, generalized Lucas matrix and Elgamal technique have been considered as a key element of system and decryption matrix is set up as $L_{\lambda}^{(-s)}H^{-1}$ constructed with combinations of terms of $ GLS(\lambda,s) $. The construction of key matrices is quite easy for authorized parties as $\lambda$ are known for both of them but it is very difficult for an adversary to obtain $\lambda$, as the adversary needs to solve a discrete logarithm problem\cite{elgamal1985public}. Further matrix construction is based on only two elements $(\lambda,s)$, so it reduces time complexity as well as space complexity of key formation and calculation of its inverse.
	In the context of attacks based on public data, one of the popular attacks is brute force attack \cite{stallings2017cryptography,stinson2005cryptography,paar2009understanding} which have been discussed here. In case of brute force attack, the adversary needs to calculate $\lambda$ which is almost impossible (discrete logarithm problem), and the next challenge for the adversary is to identify the correct key matrix out of $ |GL(\lambda)| $ matrices, where $ GL(\lambda)$ represents general linear group\cite{dummit2004abstract} of order $\lambda$. Since $
	 |GL(\lambda)| $ given by 
	\begin{eqnarray}\label{Gln}
		|GL_{\lambda}(F_{p})| = (p^{\lambda}-p^{\lambda-1})(p^{\lambda}-p^{\lambda-2}) \cdots (p^{\lambda}-p^{1})(p^{\lambda}-1)
	\end{eqnarray}
	From equation(\ref{Gln}), it is clear that security for generalized Lucas matrices $L_{\lambda}^{(s)}$ depends on $\lambda$ only, not on signature $s$. So $s$ does not compromise the security even though it is known to the adversary.
	For example, consider $ p=37 $ and $\lambda = 50$, then by equation(\ref{Gln}) total number of possible key space over $  \mathbb{F}_{37} $ is approximately $3.105\times 10^{3920}$ which is too large. And in case of $\lambda$ and/or prime p increasing, then the key space grows exponentially.
	\section{Conclusion}
	We know that key component in a cryptosystem plays an important role to keep cryptography secure and strong. Many researchers are working in the direction of developments of new cryptography and modification in existing cryptography based on properties of number theory and linear algebra.	
	
	Here, in this article, we have first proposed a generalized Lucas sequence which is constructed with the trace of $ GFM(\lambda,s) $ and its initial values are also obtained. Further, we have developed a recursive matrix whose entries are constructed from linear combinations of generalized Lucas sequences and investigated some properties like direct calculation of its inverse, product of two matrices, ...etc irrespective of matrix algebra.
	The purpose behind considering the field for a key element is the necessity of the existence of inverse. In our case, we have considered generalized Lucas matrices which  does not form a multiplicative group but it assure the existence of inverse matrix for each $GLM (\lambda,s)$ for every $\lambda(\geq 2)\in\N$, i.e. we have shown that for every integer $s$, we have matrix $K^*=L_{\lambda}^{(-s)}H^{-1}$ such that $L_{\lambda}^sK^*=K^*L_{\lambda}^s=I_{\lambda}$.

	Also, we have proposed a modified public key cryptography using Affine-Hill cipher \& Elgamal signature schemes with generalized Lucas matrices and show the implementation of Lucas matrices as a key matrix.
	
	Thus generalized Lucas matrices as key components in cryptosystem play a crucial role. It enlarges the keyspace, reduces the time complexity as well as the space complexity of key formation in cryptography. Our proposed method is based on construction with two parameters and has three digital signatures ($\lambda, s ~\& ~\text{shift vector(B)}$) which strengthen the security of the modified cryptography.
	Since $ \lambda $ is known only to both end parties(Alice and Bob) so shift vector B constructed with $ \lambda $ is also known only to Alice and Bob, thus it is practically impossible to recover $\lambda$ by anyone else as it is based on discrete logarithm problem. Hence, our proposed method is mathematically simple for authorized party and tedious for an intruder, mathematically strong and have a large keyspace.

	\subsection*{Acknowledgment}
	The first and third authors acknowledge the University Grant Commission, India for providing fellowship for this research work.
		
\bibliography{LucasCryptography}
\bibliographystyle{acm}
\end{document}